\documentclass[12pt,twoside]{article}
\usepackage{amsmath,amssymb,color,epsfig,epic,subfigure,amsfonts,hhline,graphicx}

\def\N{{\rm I\kern-.25em N}}
\def\R{{\rm I\kern-.25em R}}

\def\Q{{\rm I\kern-.25em Q}}

\def\mbN{{\mathbb N}}
\def\mbR{{\mathbb R}}
\def\mbC{{\mathbb C}}
\def\mbE{{\mathbb E}}
\def\mbK{{\mathbb K}}
\def\mbP{{\mathbb P}}

\def\mbZ{{\mathbb Z}}

\def\mcB{\mathcal{B}}

\def\mcH{\mathcal{H}}

\def\mcM{\mathcal{M}}

\def\ol{\overline}

\def\sbs{\subset}

\def\ptl{\partial}

\def\wt{\widetilde}
\def\wh{\widehat}

\def\d{\delta}    
\def\a{\alpha}    \def\b{\beta}              \def\d{\delta}
    \def\e{\varepsilon} \def\ph{\varphi}       
  \def\g{\gamma}           
\def\l{\lambda}        \def\m{\mu}        
\def\r{\rho}               
\def\p{\pi}              \def\s{\sigma}     
\def\th{\theta}          
\def\x{\xi}       
   \def\ps{\psi}

\newtheorem{thm}{Theorem}[section]
\newtheorem{lem}[thm]{Lemma}
\newtheorem{prop}[thm]{Proposition}
 {
\newtheorem{define}[thm]{Definition}
\newtheorem{rem}[thm]{Remark}
\newtheorem{ex}[thm]{Example}
}

\newenvironment{proof}{\bigskip\par\noindent{\it Proof:}}{$\square$\newline\vspace*{0.2cm}}

\makeatletter \@addtoreset{equation}{section}
\renewcommand{\thesection}{\arabic{section}}

\makeatother \makeatother
\begin{document}

\title{The Generator and Quantum Markov Semigroup
for Quantum Walks}
\author{}
\date{}
\maketitle

\date{}
\maketitle
 \baselineskip=18pt
\vspace{-3cm}
\begin{center}
%\begin{minipage}{5in}
\parbox{5in}

\noindent{ Chul Ki Ko}

\noindent {\small University College, Yonsei University\\[-0.1cm] 134
Sinchon-dong, Seodaemun-gu,
Seoul 120-749, Korea\\[-0.1cm]
e-mail: kochulki@yonsei.ac.kr} \vskip 0.5 true cm

 \noindent Hyun Jae
Yoo\footnote{Corresponding author}

\noindent {\small Department of Applied
Mathematics, Hankyong National University\\[-0.1cm] 327 Jungangro, Anseong-si, Gyeonggi-do 456-749, Korea\\[-0.1cm]
e-mail: yoohj@hknu.ac.kr} \vskip 0.2 true cm
\end{center}

{The quantum  walks in the lattice spaces are represented as
unitary evolutions. We find a generator for the evolution and
apply it to further understand the walks. We first extend the
discrete time quantum  walks to continuous time walks. Then
we construct the quantum Markov semigroup for quantum walks and
characterize it in an invariant subalgebra. In the meanwhile, we
obtain the limit distributions of the quantum  walks in
one-dimension with a proper scaling, which was obtained by Konno
by a different method. }

\vspace{0.2cm} {\bf Key words} : Quantum  walks,
Schr\"odinger approach, generator, continuous time quantum
walks, limit distribution, superposition, quantum Markov
semigroup.

\vspace{0.2cm} {\bf Mathematical subject classification (2000)} :
82B41, 60F05, 47D07.

\setlength{\hoffset}{-0.2in}

\baselineskip=18pt
\section{Introduction}
\quad Quantum walk (QW hereafter) is a quantum analogue of
classical random walk. After it was initiated by Meyer \cite{M}, it
attracted many interests and there are many works developing it in
mathematically rigorous way on the one hand and explaining
possible practical applications, e.g., in quantum computation (see
\cite{ABNVW, GJS, Ke, K1, K2, L, M}, and references therein for
more details).

QW's demonstrate non-intuitive behaviour in several ways
comparing to classical random walks. The most outstanding feature is fast
diffusing as noted by many authors: the scaling for the central
limit theory is $n$ comparing to $\sqrt{n}$ for classical random walks. It
is caused from quantum interference. The superposition in QW's is
likewise a unique phenomenon that does not exist in classical
random walks.

The aim of this paper is to further investigate the QW's by their
generators. We find the generator from an evolution map of a QW.
As applications we will first extend the discrete time QW's to
continuous time walks. We also discuss the quantum Markov
semigroup for QW's. The quantum probabilistic aspect of the QW's
has been discussed in a separate paper \cite{KY}. We remark that
there already have been studies of continuous time QW's on the
graphs \cite{FG, K3, MW, O, SJ}, but we emphasize that the
extension here is different from those. It is a natural extension
of the discrete time QW on integer lattices in the sense that it
agrees with the original discrete time QW for integer times. We
note that this concept was already appeared in \cite{KFK}. Next,
not only we construct the quantum Markov semigroup for QW's, we
also find an invariant subalgebra on which the dynamics is
completely characterized.

Our method is to use Fourier transform, so called a Schr\"odinger
approach, which was introduced by Ambains {\it et al.}
\cite{ABNVW, NV}. By it we will recover the limit distributions
for QW's which was concretely studied by Konno \cite{K1, K2} via
path integral approach.

This paper is organized as follows. In section 2, we briefly
review the QW's and find the unitary evolution map of them. Then,
we find a scaled limit distributions of QW's (Theorem
\ref{thm:asymptotics} and Proposition \ref{prop:local_intl_cond}).
In section 3, we observe a superposition phenomena for a typical
Hadamard walk. Then we find a continuous time extension. In
section 4, we discuss the quantum Markov semigroup for QW's.

\section{1-dimensional Quantum  Walks}\label{sec:1-d QW}
In this section we briefly introduce the 1-dimensional QW's. We
will see that a QW is a (discrete time) unitary evolution in a
suitably chosen Hilbert space.

\subsection{1-dimensional QW's}
We first introduce the definition of 1-dimensional quantum
walks following \cite{ABNVW, GJS, K1, K2, NV}. A quantum particle
has an intrinsic degree of freedom, called \lq\lq chirality\rq\rq.
This chirality is represented by a 2-dimensional vector: we
represent them in $\mbC^2$ and call the vectors
$\left(\begin{matrix}1\\0\end{matrix} \right)$ and
$\left(\begin{matrix}0\\1\end{matrix}\right)$ the left and right
chirality, respectively. The spatial movement of the particle is
given as follows. At time $n\in \mbN_0=\{0,1,2,\cdots\}$, the
probability amplitude of finding the particle at site $x\in \mbZ$
with chirality state being left or right is given by a
two-component vector
\begin{equation}\label{eq:amplitude}
\ps_n(x)=\left(\begin{matrix}
\ps_n(1;x)\\\ps_n(2;x)\end{matrix}\right)\in\mbC^2.
\end{equation}
After one unit of time the chirality is rotated by an a priori
given unitary matrix $U$. According to the final chirality state,
if the particle ends up with left chirality, then it moves one
step to the left, and if it ends up with right chirality, it moves
one step to the right. In order to see this dynamics more
precisely let us denote
\begin{equation}\label{eq:unitary_matrix}
U=\left(\begin{matrix} l_1&l_2\\r_1&r_2\end{matrix}\right)
\end{equation}
and define
\begin{equation}\label{eq:chiral_matrix}
L=\left(\begin{matrix}l_1&l_2\\0&0\end{matrix}\right)\text{ and
}R=\left(\begin{matrix} 0&0\\r_1&r_2\end{matrix}\right).
\end{equation}
Then the dynamics for $\ps_n=(\ps_n(x))_{x\in \mbZ}$ is given by
\begin{equation}\label{eq:evolution1}
\ps_{n+1}(x)=L\ps_n(x+1)+R\ps_n(x-1).
\end{equation}
This dynamics has been investigated by many authors. There are two
main methods to investigate it. One is so called the path integral
approach, in which the explicit probability amplitude is computed
by using a great deal of combinatorics. This method has been
extensively developed by Konno \cite{K1,K2}. In particular, Konno
obtained the scaled limit distribution of the QW very concretely.
The other method is called the Schr\"odinger approach, which uses
Fourier transform taking advantage of space-time homogeniety of
QW's. This approach was well-developed in \cite{ABNVW, GJS, KFK,
NV}. In this paper we further develop the Schr\"odinger approach
to get a unitary evolution map for the QW in a suitable Hilbert
space. Then the generator comes out naturally.
\subsection{Evolution of QW's}\label{subsec:evolution}
For each $x\in \mbZ$, let $\mcH_x:=\mbC^2$ be a copy of the
chirality space. Let
\begin{equation}\label{eq:Hilbert_space}
\mcH:=\oplus_{x\in \mbZ}\mcH_x
\end{equation}
be the direct sum Hilbert space, on which the evolution of a QW
will be developed. Notice that $\mcH$ is isomorphic to the Hilbert
spaces $l^2(\mbZ,\mbC^2)$ and $l^2(\mbZ)\otimes \mbC^2$.  For each
$x\in \mbZ$, let
\begin{equation}\label{eq:onb}
e_x(k):=\frac1{\sqrt{2\p}}e^{i xk},\quad k\in \mbK:=(-\p,\p],
\end{equation}
$\mbK$ being understood as a unit circle in $\mbR^2$. The set
$\{e_x\}_{x\in \mbZ}$ defines an orthonormal basis in $L^2(\mbK)$.
For each $k\in\mbK$, let ${\sf h}_k$ be a copy of $\mbC^2$ and let
\begin{equation}\label{eq:direct_integral}
\wh{\mcH}:=\int_{\mbK}^\oplus {\sf h}_kdk\approx
L^2(\mbK,\mbC^2)\approx L^2(\mbK)\otimes \mbC^2
\end{equation}
be the direct integral of Hilbert spaces. The Fourier transform
between $l^2(\mbZ)$ and $L^2(\mbK)$ naturally extends to a unitary
map from $\mcH$ to $\wh{\mcH}$ by
\begin{equation}\label{eq:F_transform}
\ps=\left\{\left(\begin{matrix}\ps(1;x)\\\ps(2;x)\end{matrix}\right)\right\}_{x\in\mbZ}\in\mcH\mapsto
\wh{\ps}=\left\{\left(\begin{matrix}\wh{\ps}(1;k)\\\wh{\ps}(2;k)\end{matrix}\right)\right\}_{k\in\mbK}
\in \wh{\mcH},
\end{equation}
where
\begin{equation}\label{eq:F_transform_vector}
\wh{\ps}(i;k)=\sum_{x\in \mbZ}\ps(i;x)e_x(k), \quad i=1,2.
\end{equation}

Its inverse is given by $\wh\ps\mapsto{\ps}$ with
\[
{\ps}(x)=\int_{-\p}^{\p}\frac{1}{\sqrt{2\p}}
e^{-ixk}\wh\ps(k)dk\in \mcH_x.
\]
Let us denote by $T$ the left translation  in $l^2(\mbZ)$:
\begin{equation}\label{eq:left_translation}
(Ta)(x)=a(x+1), \quad \text{for } a=(a(x))_{x\in \mbZ}.
\end{equation}
$T$ is a unitary map whose adjoint is the right translation:
\begin{equation}\label{eq:right_translation}
(T^*a)(x)=a(x-1), \quad \text{for } a=(a(x))_{x\in \mbZ}.
\end{equation}
The operator $T$ naturally extends to $\mcH=\oplus_{x\in \mbZ}
\mcH_x$ and for the sake of simplicity we use the same notation
$T$ for the extension. Given an operator ($2\times 2$ matrix)
 $B$ on $\mbC^2$, we let
\begin{equation}\label{eq:operator_direct_sum}
\wt{B}:=\oplus_{x\in \mbZ}B
\end{equation}
be the bounded direct sum operator acting on $\mcH$.

With these preparations we can rewrite the dynamics of a QW as an
evolution map in the Hilbert space $\mcH$. Notice that the
equation \eqref{eq:evolution1} is the same as
\begin{equation}\label{eq:evolution2}
\ps_{n+1}(x)=L(T\ps_n)(x)+R(T^*\ps_n)(x), \quad x\in \mbZ,
\end{equation}
which we can write in a single equation:
\begin{equation}\label{eq:evolution}
\ps_{n+1}=(\wt{L}T+\wt{R}T^*)\ps_n.
\end{equation}
It is not hard to see that the operator $\wt{L}T+\wt{R}T^*$ is a
unitary operator on $\mcH$. Thus the solution to
\eqref{eq:evolution} is easily seen to be
\begin{equation}\label{eq:solution}
\ps_n=(\wt{L}T+\wt{R}T^*)^n\ps_0.
\end{equation}
This is the time evolution of the QW that we are looking for. One may write the unitary
$\wt{L}T+\wt{R}T^*$ as $T \wt{L}+T^* \wt{R}$ by noticing $\wt {L}T=T\wt{L}$ and $\wt{R}T^*=T^*\wt{R}$, if one stresses the
order that the movement (space translation) follows the action of chirality rotation.

Now we find the evolution of the QW in a Fourier transform space.
Notice that the translation operator $T$ is represented as a
multiplication operator by $e^{-i k}$ in the Fourier transform
space. Thus, the evolution in \eqref{eq:solution} has the
representation in Fourier transform space as follows:
\begin{eqnarray}\label{eq:evolution_in_F_space}
\wh{\ps}_n(k)&=&(e^{-i  k}L+e^{i  k}R)^n\wh{\ps}_0(k)\nonumber\\
&=&\left(\begin{matrix}e^{-i  k}l_1&e^{-i  k}l_2\\e^{i  k}r_1&e^{i
k}r_2\end{matrix}\right)^n\wh{\ps}_0(k).
\end{eqnarray}
This representation has been already obtained in \cite{ABNVW, GJS,
KFK, NV}. Notice that for each $k\in \mbK$ the matrix
\begin{equation}\label{eq:U(k)}
U(k):=\left(\begin{matrix}e^{-i  k}l_1&e^{-i  k}l_2\\e^{i
k}r_1&e^{i  k}r_2\end{matrix}\right)
\end{equation}
is a unitary matrix in $\mbC^2$, and hence the evolution in
\eqref{eq:evolution_in_F_space} is again unitary in $\wh{\mcH}$,
as it should be.

The probability density to find out the particle at a site $x\in
\mbZ$ at time $n$ is simply
\begin{equation}\label{eq:prob_density}
\|\ps_n(x)\|^2=|\ps_n(1;x)|^2+|\ps_n(2;x)|^2,
\end{equation}
or it can also be given by
\begin{equation}\label{eq:prob_density_k_space}
\left\|\int_{-\p}^\p\frac1{\sqrt{2\p}}e^{-i
xk}\wh{\ps}_n(k)dk\right\|^2=\frac1{2\p}\left\{\left|\int_{-\p}^\p
e^{-i  xk}\wh{\ps}_n(1;k)dk\right|^2 +\left|\int_{-\p}^\p e^{-i
xk}\wh{\ps}_n(2;k)dk \right|^2\right\}.
\end{equation}
Konno has obtained the explicit form of the density
\eqref{eq:prob_density} by using previously mentioned path
integral approach. It uses a good deal amount of combinatorics and
the resulting formula looks rather complicated \cite{K1,K2}.
Nevertheless, by using his formula, Konno has successfully
obtained the asymptotic distributions of the scaled QW's. On the
other hand, by using the formula in
\eqref{eq:prob_density_k_space}, Ambainis {\it et al.} also
explained many properties of QW's \cite{ABNVW,NV}. In particular,
when one is interested in the asymptotic behavior of QW's it turns
out that the formula in  \eqref{eq:prob_density_k_space} is
extremely convenient because we have a nice tool so called the
method of stationary phase \cite{BO, BH}. The asymptotic behavior
of the probability amplitudes by this method was investigated by
Ambainis {\it et al.} \cite{ABNVW, NV}. In the next subsection we
will find the limit distribution of the scaled QW by computing the
limit of characteristic functions. We notice that Grimmett {\it et
al.} obtained also the weak limit of the scaled QW's by using the
method of moments in the Schr\"odinger approach \cite{GJS}. In
\cite{KFK}, Katori {\it et  al} further developed this method and
they re-established the limit distribution. The moment problem is
closely related to the interacting Fock spaces via quantum
probability theory, which we have discussed in other paper
\cite{KY}.

\subsection{Limit Distributions}\label{sub:limit_distributions}
In this subsection we study the limit distribution of the scaled
QW. Let $\{X_n^{(U;\ps_0)}\}_{n\ge 0}$ be the random variables
distributed on the integer space $\mbZ$
 according to the QW whose evolution is given by \eqref{eq:solution}. That is,
\begin{equation}\label{eq:probability}
\mbP(X_n^{(U;\ps_0)}=x)=\|\ps_n(x)\|^2.
\end{equation}
Before we state the result we notice that a multiplication by a
phase factor to $U$ does not affect the distribution of
$\{X_n^{(U;\ps_0)}\}$. Thus, for a technical reason in the proof,
we will assume that
\begin{equation}\label{eq:modulation}
\det U=1.
\end{equation}
Thereby we caution the reader that if the matrix $U$ in a given
model does not satisfy \eqref{eq:modulation}, we will first adjust
it by multiplying some phase factor so that \eqref{eq:modulation}
is satisfied.
\begin{thm}\label{thm:asymptotics}
There is a random variable $Z^{(U;\ps_0)}$ on the real line such
that in distribution
\begin{equation}\label{eq:asymptotics}
\lim_{n\to \infty}\frac{X_n^{(U;\ps_0)}}{n}=Z^{(U;\ps_0)}.
\end{equation}
If $l_1l_2r_1r_2\neq 0$, the distribution $\m^{(U;\ps_0)}$ of
$Z^{(U;\ps_0)}$ has a density function: it is supported on
$(-|l_1|, |l_1|)$ and has the form:
\begin{equation}\label{eq:density}
\r^{(U;\ps_0)}(y)=\frac{\sqrt{1-|l_1|^2}}{\p
(1-y^2)\sqrt{|l_1|^2-y^2}}g^{(U;\ps_0)}(y)
\end{equation}
with $g^{(U;\ps_0)}(y)$ being a dependent part to the initial
condition. On the other hand, if one of $l_1$ or $l_2$ is zero,
then the distribution $\mu^{(U;\psi_0)}$ is a point mass: for
$\psi_0=\left\{\left(\begin{matrix}\ps_0(1;x)\\\ps_0(2;x)\end{matrix}\right)\right\}_{x\in\mbZ}$,\\
\begin{equation}\label{eq:density2}
\m^{(U;\ps_0)}=\begin{cases} (\sum_{x\in \mbZ}|\ps_0(1;x)|^2)
\d_{-1}+(\sum_{x\in \mbZ}|\ps_0(2;x)|^2)\d_1,&\text{if }l_2=0\\
\d_0,&\text{if }l_1=0\end{cases}.
\end{equation}
\end{thm}
\begin{rem}\label{rem:density}
(a) The function $g^{(U;\ps_0)}(y)$ depends heavily on the initial
state $\ps_0$. In Proposition \ref{prop:local_intl_cond} below we
will see a concrete form of $g^{(U;\ps_0)}(y)$ for QW's that are
initially localized at the origin. The above formula was first
shown by Konno \cite{K1,K2}. Grimmett {\it et al.} also obtained
the formula for the (biased) Hadamard QW's \cite{GJS}. Katori
{\it et al.} recovered it from the method of moments \cite{KFK}. Recently Ahlbrecht {\it et al.} discussed the asymptotic behaviour or QW's by using a perturbative method \cite{AVWW}.

(b) In relevance with the limit theory, we would like to mention some recent results. Sunada and Tate investigated the
limit theory of the quantum walk (starting at one point, say the origin) much more closely dividing the region into three areas: allowed region
(inside the interval $(-|l_1|, |l_1|)$), around the wall ($|x|\sim \pm |l_1|$), and hidden region ($|l_1|<|x|<1$). In particular, for the hidden region,
 they obtained the large deviation principle, i.e., the probability in the hidden region decreases exponentially with a concrete rate function. See \cite{ST} for the details.
 In \cite{Ma}, Machida investigated that by allowing various initial conditions, in the limit we can recover some of the well known distributions
 such as semicircular law, arcsine law, Gaussian, and uniform distributions.
\end{rem}
The proof of Theorem \ref{thm:asymptotics} will be given in the
Appendix. Although it was shown already, our Shr\"odinger approach
should be a good contrast to the path integral approach. As
mentioned, the method of stationary phase plays the key role for
asymptotics of the integral of rapidly varying functions.

Next we consider the situation that the particle is initially
located at the origin. We will get more concrete form of the limit
density function.
\begin{prop}\label{prop:local_intl_cond}
Suppose that the initial condition is a qubit state
$\left(\begin{matrix}a\\b\end{matrix}\right)$, $a,b\in \mbC$,
$|a|^2+|b|^2=1$, located at the origin. Then the density of the
limit distribution in Theorem \ref{thm:asymptotics} in the case
$l_1l_2r_1r_2\neq 0$ is given by the following formula.
\[
\r^{(U;\ps_0)}(y)=\frac{\sqrt{1-|l_1|^2}}{\p
(1-y^2)\sqrt{|l_1|^2-y^2}}\left(1-\b^{(U;\ps_0)}y\right)1_{(-|l_1|,|l_1|)}(y)
\]
with
\[
\b^{(U;\ps_0)}=|a|^2-|b|^2+\frac{\ol{
l_1}l_2\ol{a}b+l_1\ol{l_2}a\ol{b}}{|l_1|^2}.
\]
\end{prop}
\begin{rem}\label{rem:Konno}
The formula in Proposition \ref{prop:local_intl_cond} is exactly
what Konno obtained by the path integral approach \cite{K1,K2}.
\end{rem}
The proof of Proposition \ref{prop:local_intl_cond} will also be
given in the Appendix.

\section{Continuous Time QW's}\label{sec:continuous time}
In this section we extend the discrete time QW's to continuous
time QW's. It is done from our development in Section 2 and we
remark that it is a different kind of version for continuous time
QW's from those appearing in the literature \cite{K3, MW, O}. As
we have seen in the last section, the distribution of QW's
depends heavily on the initial condition. In particular, the QW's
reveal the superposition of states. In the next subsection we will
see the superposition phenomena in the simplest case of Hadamard
walk.
\subsection{Superposition of QW's}\label{subsec:superposition}
Let us consider the Hadamard QW with the unitary matrix for the
rotation of chirality given by
\begin{equation}\label{eq:Hadamard_switched}
U=\frac1{\sqrt{2}}\left(\begin{matrix}1&-1\\1&1\end{matrix}\right).
\end{equation}
We notice here that we have changed the rows of the matrix from
the usual Hadamard matrix. It is just to make  $\det U=1$ and it
only makes the exchange of left and right movements of the quantum
walker. We will consider for the initial conditions not only the
case that the walker starts at the origin but also the case that
it is spatially distributed.

\begin{figure}[h]
\begin{center}
\includegraphics[width=0.4\textwidth]{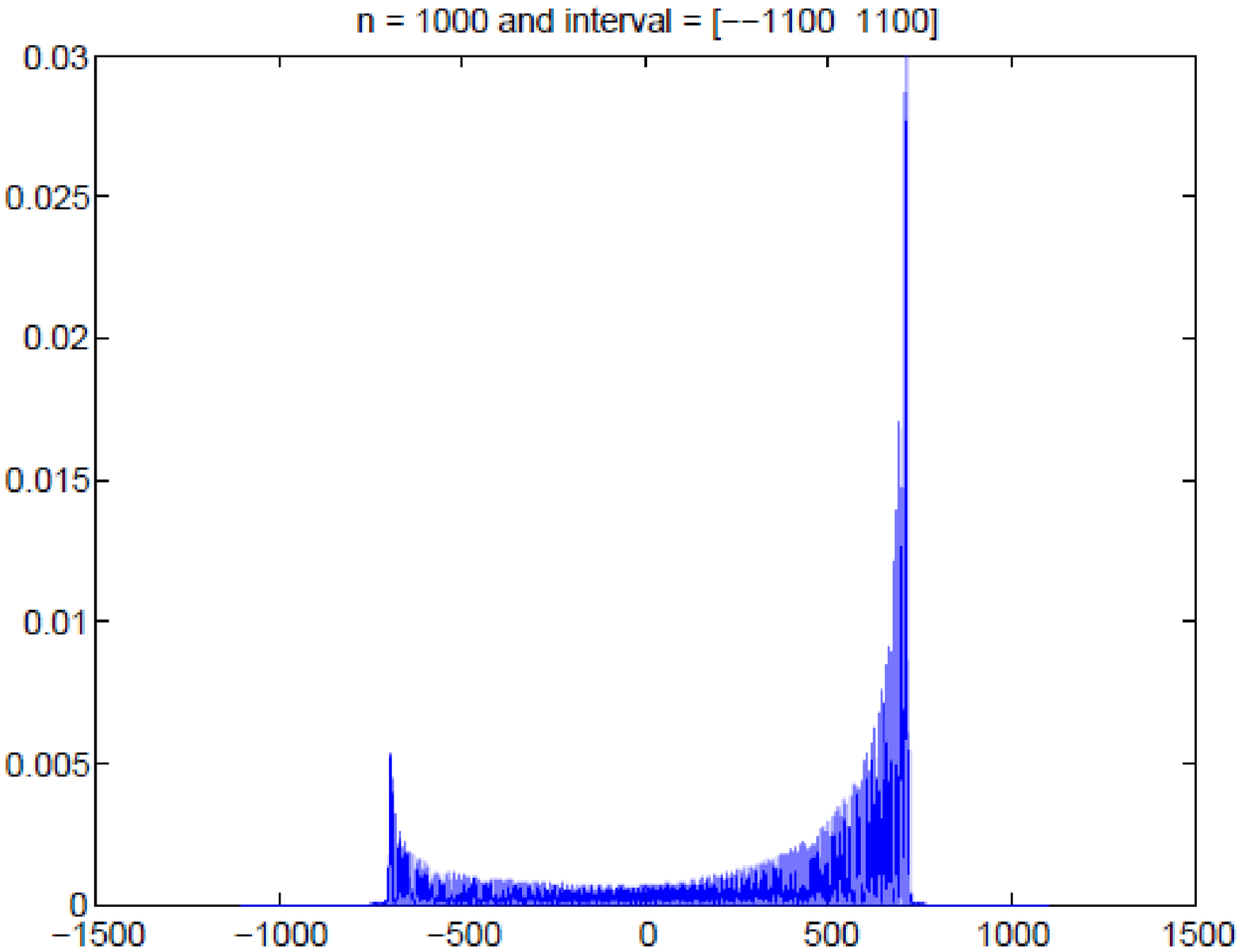}~~~~~~~~\includegraphics[width=0.4\textwidth]{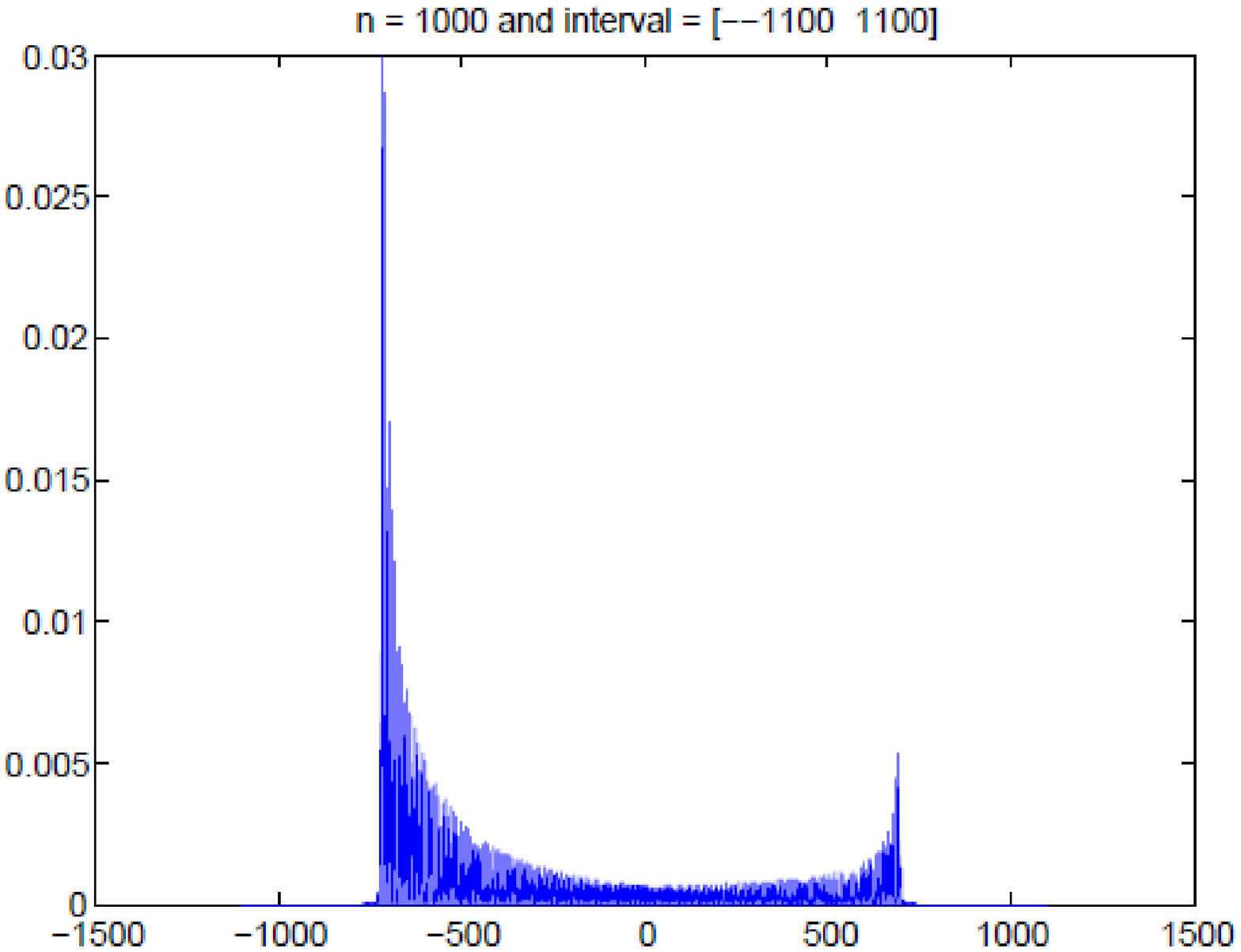}

\text{Figure 3.1} \hspace*{3cm}\text{Figure 3.2}

\includegraphics[width=0.385\textwidth]{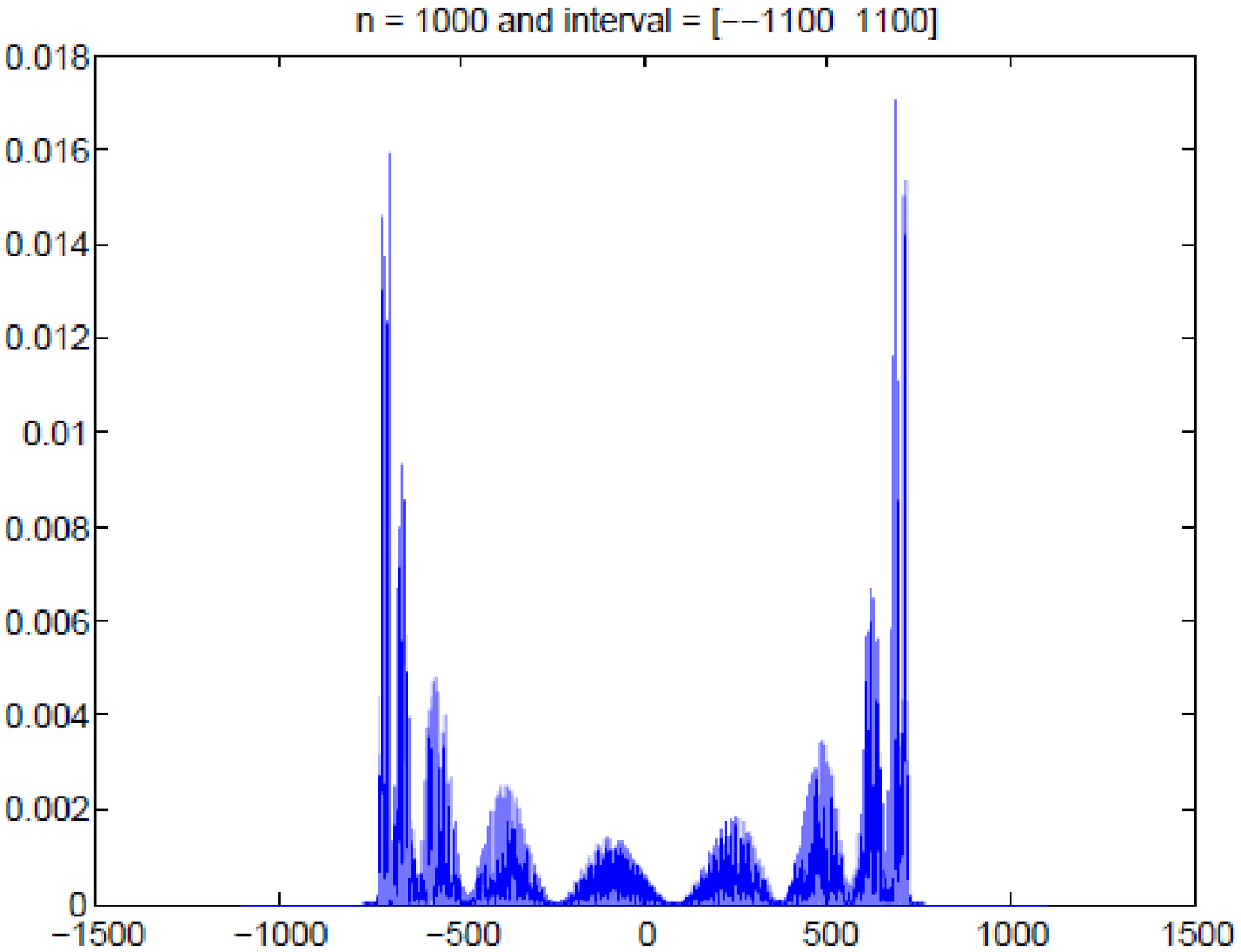}~~~~~~~~\includegraphics[width=0.4\textwidth]{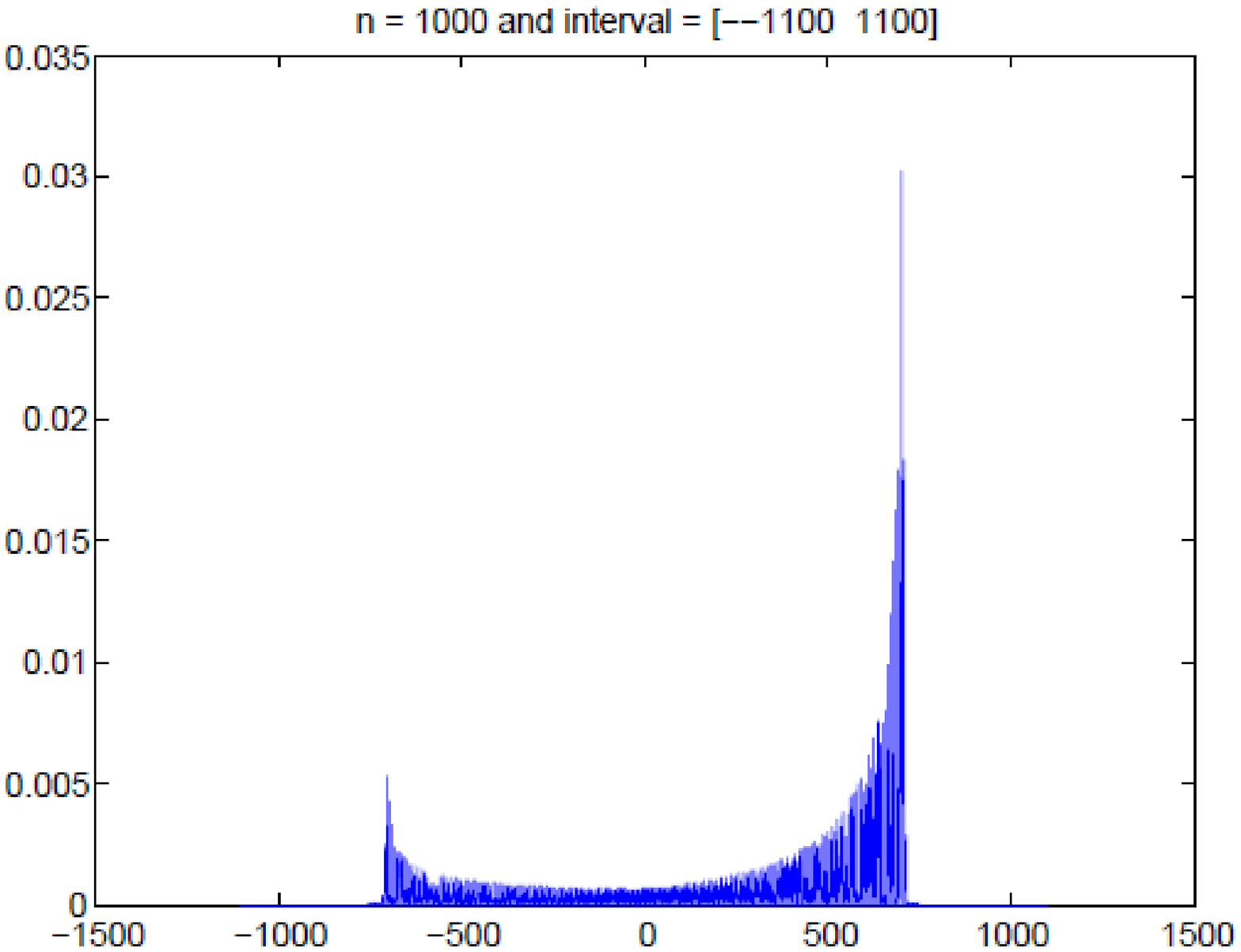}

\hspace*{0.5cm}\text{Figure 3.3} \hspace*{3cm}\text{Figure 3.4}
\end{center}
\end{figure}

Figure 3.1 shows the spatial distribution of the QW at time
$n=1000$ starting at the point $x=10$ with initial qubit state
$\left(\begin{matrix}0\\1\end{matrix}\right)$, i.e.,
$\ps_0=\left\{\left(\begin{matrix}0\\1\end{matrix}\right)\d_{10}(x)\right\}_{x\in\mbZ}$,
or $\wh
\ps_0(k)=\frac1{\sqrt{2\p}}\left(\begin{matrix}0\\e^{10ik}\end{matrix}\right)$.
Similarly Figure 3.2 shows the distribution at $n=1000$ with
$\ps_0=\left\{\left(\begin{matrix}1\\0\end{matrix}\right)\d_{-10}(x)\right\}_{x\in\mbZ}$.
Figure 3.3  shows the distribution at $n=1000$ with
$\ps_0=\left\{\left(\begin{matrix}0\\\frac{1}{\sqrt{2}}\end{matrix}\right)\d_{10}(x)+\left(\begin{matrix}\frac{1}{\sqrt{2}}\\0\end{matrix}
\right) \d_{-10}(x)\right\}_{x\in\mbZ}$, the mixture of the
previous two examples.  It shows the superposition of the QW.
Finally Figure 3.4 shows the distribution at $n=1000$ for
$\ps_0=\left\{\left(\begin{matrix}0\\\frac{1}{\sqrt{2}}\end{matrix}\right)\d_{0}(x)+\left(\begin{matrix}\frac{1}{\sqrt{2}}\\0\end{matrix}\right)
\d_{0}(x)\right\}_{x\in\mbZ}$. We see that if it were the
classical random walk, then the distribution for the initial condition in
Figure 3.3 would be the mean of the distributions of the Figure
3.1 and 3.2. But the distribution for the QW is totally different
from this behavior and the result in Figure 3.3 shows that in
QW's the walks have interference to each other, like in a two
slit experiment in quantum mechanics. Figure 3.4 shows that it is
still different from the behavior of the QW who starts at the
origin with mixed qubit state of the two walkers of Figure 3.3.
Notice that the two walkers positioned at $x=10$ and $x=-10$ might
be viewed as positioned \lq\lq almost\rq\rq
 at the origin
if one looks at them from a \lq\lq long\rq\rq distance of size
$1000$. But the results of Figure 3.3 and 3.4 show that it is
different from the intuition.

\subsection{Continuous Time QW's}\label{subsec:continuous time}
We recall the evolution of QW in \eqref{eq:evolution_in_F_space}:
\[
\wh\ps_n(k)=U(k)^n\wh\ps_0(k),
\]
where
\begin{equation}\label{eq:unitary_again}
U(k)=\left(\begin{matrix}e^{-ik}l_1&e^{-ik}l_2\\e^{ik}r_1&e^{ik}r_2\end{matrix}\right).
\end{equation}
By \eqref{eq:diagonalization} the unitary matrix $U(k)$ is
diagonalized as
\[
U(k)=S(k-\th_1)\left(\begin{matrix}e^{i \g(k-\th_1)}&0\\0&e^{-i
\g(k-\th_1)}\end{matrix}\right)S(k-\th_1)^{-1}.
\]
Thus we can rewrite it as
\begin{equation}\label{eq:U(k)_canonical}
U(k)=e^{iH(k)},
\end{equation}
where $H(k)$ is a self-adjoint operator defined by
\begin{equation}\label{eq:generator}
H(k)=S(k-\th_1)\left(\begin{matrix} \g(k-\th_1)&0\\0&
-\g(k-\th_1)\end{matrix}\right)S(k-\th_1)^{-1}.
\end{equation}
The evolution of QW can now be denoted by
\begin{equation}\label{eq:evolution_1}
\wh\ps_n(k)=e^{inH(k)}\wh\ps_0(k).
\end{equation}
Now it is strightforward to extend the QW to a continuous time
QW:
\begin{define}\label{def:QW_cont_time}
Let $U$ be a $2\times 2$ unitary matrix. The continuous time QW
on $\mbZ$ is defined by the unitary evolution (in Fourier space)
defined by
\begin{equation}\label{eq:evolution_in_cont_time}
\wh\ps_t(k)=e^{itH(k)}\wh\ps_0(k),
\end{equation}
where $H(k)$ is the self-adjoint operator given in
\eqref{eq:generator}.
\end{define}
\begin{rem}\label{rem:difference}
(a) As mentioned before, this continuous extension of QW is different
from the usual ones on the graphs, where the generator comes from
the discrete Laplacian. Moreover, the intrinsic chiral state is
not concerned in those models, but here the continuous time QW
has still the chiral states.

(b) From \eqref{eq:evolution_in_cont_time}, one notices that the quantum walk unitary evolution satisfies the Schr\"odinger equation
(in the Fourier transform space $\wh{\mcH}=L^2(\mbK,\mbC^2)$):
\begin{equation}\label{eq:Schrodinger}
\frac{\ptl \wh{\ps}_t}{\ptl t}=i  {H}\wh{\ps}_t,\quad \wh{\ps}_t\in \wh{\mcH},
\end{equation}
where the Hamiltonian operator $ {H}$ is given by
\begin{equation}\label{eq:generator_global}
 {H}=\int_{\mbK}^{\oplus}H(k)dk.
\end{equation}
If we pull back the equation in the real Hilbert space $\mcH=l^2(\mbZ,\mbC^2)$, then it is written as
\begin{equation}\label{eq:Schrodinger_real_space}
\frac{\ptl  {\ps}_t}{\ptl t}=i K {\ps}_t,\quad  {\ps}_t\in  {\mcH},
\end{equation}
where the Hamiltonian operator $ K$ works as
\[
 (K\psi)(x)=\frac1{\sqrt{2\p}}\int_{-\p}^\p e^{-ixk}H(k)\wh{\ps}(k)dk,\quad \ps\in \mcH,
\]
where $\wh{\ps}$ is the Fourier transform of $\ps$.
\end{rem}
\begin{ex} We consider again the Hardamard walk of the previous subsection but in the continuous time. We take the initial condition of Figure 3.3,
i.e.,
$\ps_0=\left\{\left(\begin{matrix}0\\\frac{1}{\sqrt{2}}\end{matrix}\right)\d_{10}(x)+\left(\begin{matrix}\frac{1}{\sqrt{2}}\\0\end{matrix}\right)
\d_{-10}(x)\right\}_{x\in\mbZ}$, or $\wh
\ps_0(k)=\frac1{2\sqrt{\p}}\left(\begin{matrix}e^{-10ik}\\e^{10ik}\end{matrix}\right)$.
The following figures show a series of snapshots of the
distribution of $X_t^{(U;\ps_0)}$ at times $t=99.25$, $99.5$,
$99.75$, and $100$.

\begin{figure}[h]
\begin{center}
~~~~~~\includegraphics[width=0.8\textwidth]{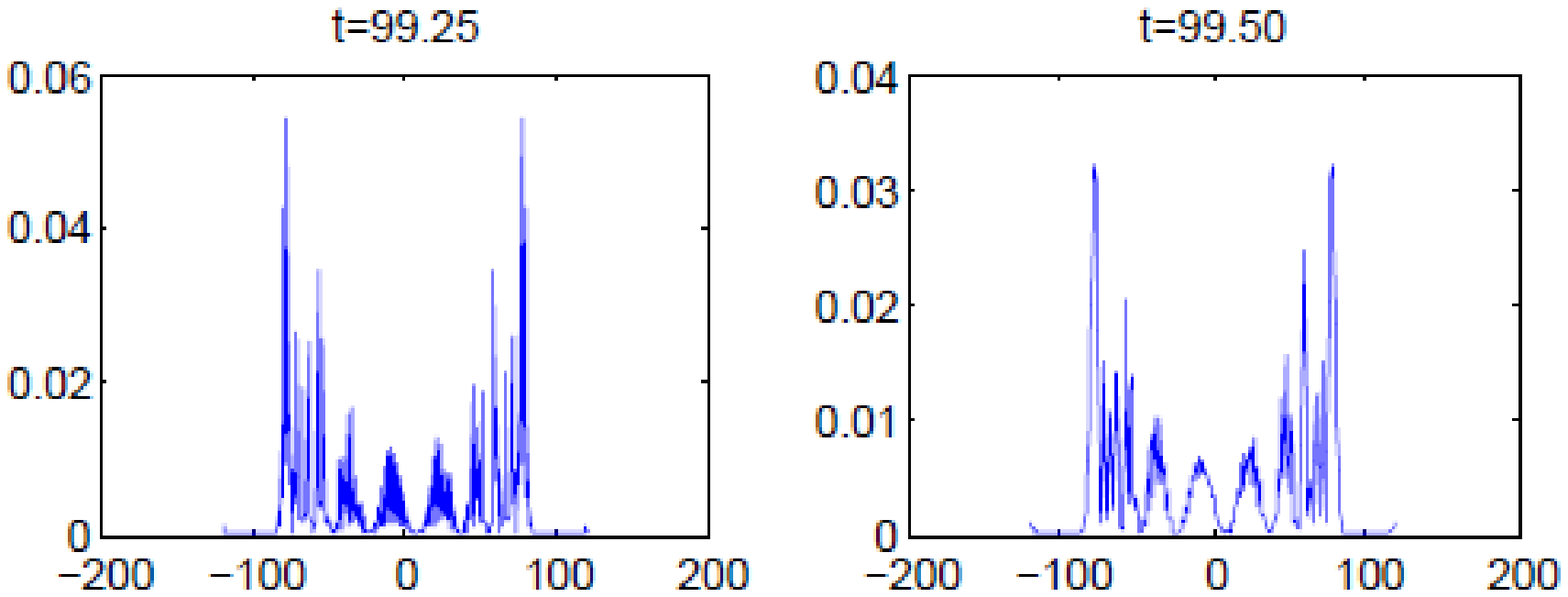}\\

%~~~~~~~\includegraphics[width=0.8\textwidth]{Fig3-5-2.eps}\\

%\hspace*{0.5cm}\text{Figure 3.5}
\end{center}
\end{figure}

\begin{figure}[h]
\begin{center}
%~~~~~~\includegraphics[width=0.8\textwidth]{Fig3-5-1.eps}\\

~~~~~~~\includegraphics[width=0.8\textwidth]{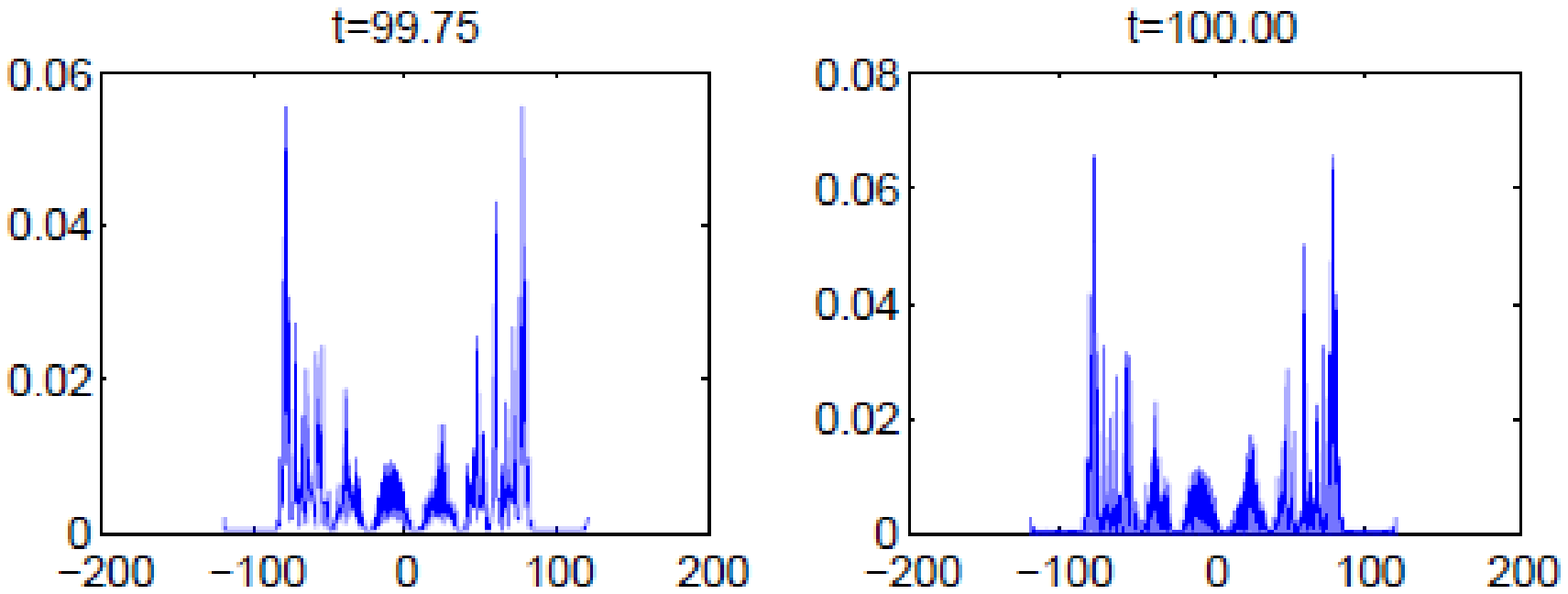}\\

\hspace*{0.5cm}\text{Figure 3.5}
\end{center}
\end{figure}

\end{ex}

\vspace{2cm}
\section{Quantum Markov Semigroup for QW's}\label{sec:QMS}
In this section we study the quantum Markov semigroup \cite{Pa}
associated to the continuous time QW's. The notion of a quantum Markov semigroup arose to describe the irreversible evolution of an open quantum system. A quantum Markov semigroup is a semigroup of completely positive, identity preserving, normal linear maps on the algebra of all bounded linear operators on a Hilbert space. Here we restrict ourselves to the evolution of observables in a closed quantum system. For the details, we refer to \cite{F} and references therein.

It turns out to be
convenient to work on the Fourier transform Hilbert space $\wh
\mcH=\int_\mbK^{\oplus}{\sf h}_kdk$, where ${\sf h}_k$ is a copy
of $\mbC^2$ for each $k\in \mbK=(-1,1]$, considered as a unit
circle in $\mbR^2$. Let $\mcM\sbs \mcB(\wh \mcH)$ be a Banach
subalgebra consisting of the operators
\begin{equation}\label{eq:subalgebra}
A:=\int_\mbK^\oplus A(k)dk\in \mcM,
\end{equation}
where $A(k)$ is a $2\times 2$ matrix for each $k\in \mbK$ and they
satisfy
\[
\sup_k\|A(k)\|<\infty.
\]
Given a unitary matrix
$U=\left(\begin{matrix}l_1&l_2\\r_1&r_2\end{matrix}\right)$,
recall the unitary matrix $U(k)$ in \eqref{eq:unitary_again}.
Notice that it defines a unitary operator on $\wh \mcH$, belonging
to $\mcM$, via the form $\int_\mbK^\oplus U(k)dk$ in the
representation of \eqref{eq:subalgebra}.  Recall the operator
$H(k)$ in \eqref{eq:generator}. By taking normalized eigenvectors
of $U(k)$ we can take $S(k)$ in \eqref{eq:generator} as a unitary
operator (see \eqref{eq:eigenvectors}):
\begin{equation}\label{eq:S(k)_unitary}
S(k)=\left(\begin{matrix}\frac1{\sqrt{1+|\a_+(k)|^2}}&\frac1{\sqrt{1+|\a_-(k)|^2}}\\
\frac{\a_+(k)}{\sqrt{1+|\a_+(k)|^2}}&\frac{\a_-(k)}{\sqrt{1+|\a_-(k)|^2}}\end{matrix}\right),
\end{equation}
where
\begin{equation}\label{eq:alpha's1}
\a_{\pm}(k)=i e^{i(k+\th_1-\th_2)}\left(|l_1|/|l_2|\,\sin k\pm
\sqrt{1+\left(|l_1|/|l_2|\,\sin k\right)^2}\right).
\end{equation}
In the above $\th_2\in \mbK$ is such that $l_2=|l_2|e^{i\th_2}$
and we have used the relation $\cos\g(k)=|l_1|\cos k$. Then $H(k)$
is given by
\begin{equation}\label{eq:H(k)_with_unitary}
H(k)=S(k-\th_1)\left(\begin{matrix}\g(k-\th_1)&0\\0&-\g(k-\th_1)\end{matrix}\right)S(k-\th_1)^*.
\end{equation}
Because $\cos^{-1}|l_1|\le \g(k)\le \p-\cos^{-1}|l_1|$ uniformly
for $k\in \mbK$, the operator norm $\|H(k)\|$ (as an operator on
$\mbC^2$) is bounded by $\p-\cos^{-1}|l_1|$ uniformly for $k\in
\mbK$. Thus the self-adjoint operator $H:=\int_\mbK^\oplus H(k)dk$
is a bounded operator on $\wh \mcH$ and belongs to $\mcM$. We
define a semigroup $V_t$ on $\mcB(\wh \mcH)$ by
\begin{equation}\label{eq:semigroup}
V_t(A):=e^{itH}Ae^{-itH}, \quad A\in \mcB(\wh \mcH).
\end{equation}
Notice that $V_t$ has the representation
\begin{equation}\label{eq:representation}
V_t(A)=e^{t\mathcal{L}}(A),
\end{equation}
where the generator $\mathcal{L}\in \mcB(\wh \mcH)$ is defined by
\begin{equation}\label{eq:semigroup_generator}
\mathcal{L}(A):=i[H,A].
\end{equation}
By the way that the operator $H$ is defined, it is clear that
$V_t$ leaves the subalgebra $\mcM$ invariant. Moreover, if $A\in
\mcM$ is represented by $A=\int_\mbK^\oplus A(k)dk$, then
\begin{equation}\label{eq:ms_on_m}
V_t(A)=\int_\mbK^\oplus V_{k,t}(A(k))dk,
\end{equation}
where
\begin{equation}\label{eq:ms_local}
V_{k,t}(A(k))=e^{itH(k)}A(k)e^{-itH(k)}=e^{t\mathcal{L}_k}(A(k)),
\end{equation}
with the local generator $\mathcal{L}_k$ defined by
\begin{equation}\label{eq:local_generator}
\mathcal{L}_k(A(k))=i[H(k),A(k)].
\end{equation}
The semigroup $\{V_t\}_{t\ge 0}$ is a quantum Markov semigroup on
$\mcB(\wh \mcH)$ \cite{Pa}. In particular it preserves the
identity and positivity. Our main purpose in this section is to
characterize the action of the semigroup $\{V_t\}_{t\ge 0}$ on the
invariant subalgebra. For it let us recall the Pauli matrices:
 $$\sigma_0=\left(\begin{matrix}1&0\\0&1\end{matrix}\right),
~\sigma_1=\left(\begin{matrix}0&1\\1&0\end{matrix}\right),
~\sigma_2=\left(\begin{matrix}0&-i\\i&0\end{matrix}\right),
~\sigma_3=\left(\begin{matrix}1&0\\0&-1\end{matrix}\right).$$
\begin{thm}\label{thm:action_of_semigroup}
For each $k\in \mbK$, there is a $3\times 3$ unitary matrix $W(k)$ such that by defining $C(k):=W(k)\left(\begin{matrix}0&0&0\\0&2\g(k)&0\\
0&0&-2\g(k)\end{matrix}\right)W(k)^*$, we have
\[
V_{k,t}(\s_0)=\s_0, \text{ and }
\left(\begin{matrix}V_{k,t}(\s_1)\\V_{k,t}(\s_2)\\V_{k,t}(\s_3)\end{matrix}\right)=e^{iC(k-\th_1)t}
\left(\begin{matrix}\s_1\\ \s_2\\ \s_3\end{matrix}\right).
\]
Therefore, for each $A\in \mcM$ of the form in
\eqref{eq:subalgebra} we have
\[
V_t(A)=\int_\mbK^\oplus \sum_{l=0}^3a_l(k)V_{k,t}(\s_l)dk,
\]
where the coefficients are such that $A(k)=\sum_{l=0}^3a_l(k)\s_l$
for each $k\in \mbK$.
\end{thm}
 \begin{proof} By direct computation, we can rewrite $H(k)$ as
\begin{eqnarray}\label{eq:local_generator_pauli}
H(k)&=&\g(k-\th_1)S(k-\th_1)\s_3S(k-\th_1)^*\nonumber\\
&=&\g(k-\th_1)\sum_{l=1}^3h_l(k-\th_1)\s_l,
\end{eqnarray}
with
\begin{eqnarray}\label{eq:H(k)_coefficients}
h_1(k)&=& \frac1{\sqrt{1+\left(|l_1|/|l_2|\,\sin k\right)^2}}(-\sin(k+\th_1-\th_2)),\nonumber\\
h_2(k)&=&\frac1{\sqrt{1+\left(|l_1|/|l_2|\,\cos k\right)^2}}(\cos(k+\th_1-\th_2)),\\
h_3(k)&=&\frac1{\sqrt{1+\left(|l_1|/|l_2|\,\sin
k\right)^2}}\left(-|l_1|/|l_2|\,\sin k\right).\nonumber
\end{eqnarray}
Notice that
\[
\frac{d}{dt}V_{k,t}(B)=V_{k,t}(\mathcal{L}_k(B))=iV_{k,t}([H(k),B])
\]
for all $2\times 2$ matrix $B$. From this and
\eqref{eq:local_generator_pauli}, and by using the commutation
relations of Pauli matrices, we have
\begin{eqnarray}\label{eq:local_evolution}
\frac{d}{dt}V_{k,t}(\s_0)&=&0,\nonumber\\
\frac{d}{dt}\left(\begin{matrix}V_{k,t}(\s_1)\\V_{k,t}(\s_2)\\V_{k,t}(\s_3)\end{matrix}\right)&=&2\left(\begin{matrix}h_1(k-\th_1)\\
h_2(k-\th_2)\\h_3(k-\th_1)\end{matrix}\right)\times
\left(\begin{matrix}V_{k,t}(\s_1)\\V_{k,t}(\s_2)\\V_{k,t}(\s_3)\end{matrix}\right),
\end{eqnarray}
where the product in the second line means the vector product of
three dimensional vectors. It is easy to solve the linear equation
\eqref{eq:local_evolution}:
\begin{align}\label{eq:local_solution}
V_{k,t}(\s_0)&=\s_0\nonumber\\
\left(\begin{matrix}V_{k,t}(\s_1)\\V_{k,t}(\s_2)\\V_{k,t}(\s_3)\end{matrix}\right)&=W(k-\th_1)\left(\begin{matrix}1&0&0\\
0&e^{2\g(k-\th_1)it}&0\\0&0&e^{-2\g(k-\th_1)it}\end{matrix}\right)W(k-\th_1)^{*}\left(\begin{matrix}\s_1\\\s_2\\\s_3\end{matrix}\right),
\end{align}
where $W(k)$ is a $3\times 3$ matrix whose columns are the
normalized eigenvectors of the matrix
\[
\left(\begin{matrix}0&-2h_3(k)&2h_2(k)\\2h_3(k)&0&-2h_1(k)\\-2h_2(k)&2h_1(k)&0\end{matrix}\right),
\]
whose eigenvalues are $0$, $\pm 2\g(k)i$. Now let
$A=\int_\mbK^\oplus A(k)dk\in \mcM$. Since the Pauli matrices
together with the identity form a basis of the algebra of $2\times
2$ matrices there are constants $a_l(k)$, $l=0,1,2,3$, such that
$A(k)=\sum_{l=0}^3a_l(k)\s_l$ for each $k\in \mbK$. Thus the
evolution of $A$ under $V_t$ is given by
\begin{eqnarray}\label{eq:evolution_on_m}
V_t(A)&=&\int_\mbK^\oplus V_{k,t}(A(k))dk\nonumber\\
&=&\int_\mbK^\oplus \sum_{l=0}^3 a_l(k) V_{k,t}(\s_l)dk,
\end{eqnarray}
with $V_{k,t}(\s_l)$, $l=0,1,2,3$, being given in
\eqref{eq:local_solution}. It completely characterizes the action
of the quantum Markov semigroup on $\mcM$.
 \end{proof}

\noindent\textbf{Acknowledgments}

We thank the anonymous referee for valuable comments. We are grateful to Boyoon Seo for helping us with the graphs. \\

\renewcommand{\thesection}{\Alph{section}}
\setcounter{section}{0}
\section{Appendix: Limit Distributions}
In this appendix, we will prove Theorem \ref{thm:asymptotics} and
Proposition \ref{prop:local_intl_cond} for the limit distributions
of 1-dimensional QW's. We start with the case $l_1l_2r_1r_2\neq
0$. The key idea is to diagonalize the matrix $U(k)$ defined in
\eqref{eq:U(k)}. Recall the unitary matrix
$U=\left(\begin{matrix}l_1&l_2\\r_1&r_2\end{matrix}\right)$. By
\eqref{eq:modulation}, we have the relations:
\begin{eqnarray}\label{eq:matrix_components}
&|l_1|^2+|r_1|^2=|l_2|^2+|r_2|^2=|l_1|^2+|l_2|^2=|r_1|^2+|r_2|^2=1;\nonumber\\
&r_1=-\ol{ l_2}, \quad r_2=\ol{ l_1}.
\end{eqnarray}
Let $\th_1\in \mbK$ be the unique number satisfying
\begin{equation}\label{eq:argument_of_l1}
l_1=|l_1|e^{i\th_1}.
\end{equation}
Then the characteristic equation  for $U(k)$ reads:
\begin{equation}\label{eq:characteristic_equation}
\l^2-2|l_1|\cos(k-\th_1)\l+1=0.
\end{equation}
Let $\g(k)$ be the nonnegative symmetric function defined on
$\mbK=(-\p,\p]$ such that
\begin{equation}\label{eq:gamma(k)}
\cos \g(k)=|l_1|\cos k,\quad k\in \mbK.
\end{equation}
In the sequel $\g(k)$ is also naturally understood as a periodic
function of period $2\p$ defined on $\mbR$. Then the solutions to
\eqref{eq:characteristic_equation}, i.e., the eigenvalues of
$U(k)$ are
\begin{equation}\label{eq:eigenvalues}
\l_+(k):=e^{i\g(k-\th_1)}\text{ and }\l_-(k):=e^{-i\g(k-\th_1)}.
\end{equation}
The corresponding (unnormalized) eigenvectors are:
\begin{eqnarray}\label{eq:eigenvectors}
e_+(k-\th_1)&\equiv&\left(\begin{matrix}u_+(k-\th_1)\\v_+(k-\th_1)\end{matrix}\right):=\left(\begin{matrix}e^{-i(k-\th_1)}\nonumber\\
-\frac{e^{i \th_1}}{l_2}\left(|l_1|e^{-i (k-\th_1)}-e^{i\g(k-\th_1)}\right)\end{matrix}\right),\\
e_-(k-\th_1)&\equiv&\left(\begin{matrix}u_-(k-\th_1)\\v_-(k-\th_1)\end{matrix}\right):=\left(\begin{matrix}e^{-i(k-\th_1)}\\
-\frac{e^{i \th_1}}{l_2}\left(|l_1|e^{-i (k-\th_1)}-e^{-i
\g(k-\th_1)}\right)\end{matrix}\right).
\end{eqnarray}
Then $U(k)$ is diagonalized as
\begin{equation}\label{eq:diagonalization}
U(k)=S(k-\th_1)\left(\begin{matrix}e^{i \g(k-\th_1)}&0\\0&e^{-i
\g(k-\th_1)}\end{matrix}\right)S(k-\th_1)^{-1},
\end{equation}
where $S(k-\th_1)$ is the matrix whose columns are $e_+(k-\th_1)$
and $e_-(k-\th_1)$. The solution $\wh{\ps}_n(k)$ in
\eqref{eq:evolution_in_F_space} then becomes
\begin{equation}\label{eq:solution_in_k_space}
\wh{\ps}_n(k)=S(k-\th_1)\left(\begin{matrix}e^{i
n\g(k-\th_1)}&0\\0&e^{-i
n\g(k-\th_1)}\end{matrix}\right)S(k-\th_1)^{-1}\wh{\ps}_0(k).
\end{equation}
In order to get the asymptotic limit \eqref{eq:asymptotics}, we
use the method of stationary phase, which we state as a lemma (see
\cite{BO,BH} for more details.).
\begin{lem}(\cite[p220]{BH})\label{lem:stationary_phase}
Suppose that $f\in C[a,b]$ and $\a\in C^2[a,b]$ with $\a$ real.
Consider the integral of the form:
\begin{equation}\label{eq:fastly_varying_integral}
I(n):=\int_a^b\exp\{in\a(t)\}f(t)dt.
\end{equation}
 Suppose further that $\a'(c)=0$ in a unique point $c\in [a,b]$ and $\a''(c)\neq 0$. Then as $n\to \infty$,
we have the asymptotic behavior of $I(n)$:
\begin{equation}\label{eq:stationary_phase}
I(n)=\exp\{in\a(c)\}f(c)\sqrt{\frac2{n|\a''(c)|}}\exp\left\{\frac{i\p
\m}{4}\right\}+o(n^{-1/2}),
\end{equation}
where $\m=\text{sign}\,\a''(c)$.
\end{lem}

{\it Proof of  Theorem \ref{thm:asymptotics}.} The case
$l_1l_2r_1r_2\neq 0$. We compute the characteristic function of
$X_n^{(U;\ps_0)}/n$:
\begin{equation}\label{eq:characteristic_function}
\ph_n^{(U;\ps_0)}(\x):=\mbE\left[e^{i\x X_n^{(U;\ps_0)}/n}\right].
\end{equation}
By using \eqref{eq:prob_density_k_space}, \eqref{eq:eigenvectors},
\eqref{eq:solution_in_k_space}, and by a translation by $\th_1$ in
the integral, we get
\begin{eqnarray}\label{eq:characteristic_function2}
\ph_n^{(U;\ps_0)}(\x)&=& \sum_{x\in \mbZ}e^{i\x x/n}
\left\{\left|\int_{-\p}^\p\frac1{\sqrt{2\p}}e^{-i
xk}\left(l_+(k)e^{in \g(k)}+l_-(k)e^{-in
\g(k)}\right)dk\right|^2\right.
\nonumber\\
& & +\left. \left|\int_{-\p}^\p\frac1{\sqrt{2\p}}e^{-i
xk}\left(m_+(k)e^{in \g(k)}+m_-(k)e^{-in
\g(k)}\right)dk\right|^2\right\},
\end{eqnarray}
where
\begin{eqnarray}\label{eq:l(k)}
&&l_+(k)=u_+(k)\left\langle\left(\begin{matrix}1\\0\end{matrix}\right),S(k)^{-1}\wh{\ps}_0(k+\th_1)\right\rangle,\\
&&l_-(k)=u_-(k)\left\langle\left(\begin{matrix}0\\1\end{matrix}\right),S(k)^{-1}\wh{\ps}_0(k+\th_1)\right\rangle,\nonumber
\end{eqnarray}
and
\begin{eqnarray}\label{eq:m(k)}
&&m_+(k)=v_+(k)\left\langle\left(\begin{matrix}1\\0\end{matrix}\right),S(k)^{-1}\wh{\ps}_0(k+\th_1)\right\rangle,\\
&&m_-(k)=v_-(k)\left\langle\left(\begin{matrix}0\\1\end{matrix}\right),S(k)^{-1}\wh{\ps}_0(k+\th_1)\right\rangle.\nonumber
\end{eqnarray}
We estimate the asymptotic integrals separately. For that, define
\begin{equation}\label{eq:I(n)}
I_{\pm}(n):=\int_{-\p}^\p\frac1{\sqrt{2\p}}e^{-i
xk}\left(l_{\pm}(k)e^{\pm i n\g(k)}\right)dk
\end{equation}
\begin{equation}\label{eq:J(n)}
J_{\pm}(n):=\int_{-\p}^\p\frac1{\sqrt{2\p}}e^{-i
xk}\left(m_{\pm}(k)e^{\pm i n\g(k)}\right)dk.\nonumber
\end{equation}
In the sum over $x\in \mbZ$ in
\eqref{eq:characteristic_function2}, we find the contribution that
gives
\begin{equation}\label{eq:weights}
\frac{x}{n}=y
\end{equation}
for a constant $y\ge 0$. The case $y<0$ is similar. Then the
integral $I_+(n)$ is rewritten as
\begin{equation}\label{eq:I+(n)}
I_+(n)=\int_{-\p}^\p e^{in(\g(k)-yk)}\frac1{\sqrt{2\p}}l_+(k)dk.
\end{equation}
In order to use Lemma \ref{lem:stationary_phase} we let
\begin{equation}\label{eq:alpha(k)}
\a(k):=\g(k)-yk.
\end{equation}
Then by definition of $\g(k)$ in \eqref{eq:gamma(k)} we see that
at two points $c_1(y)$ and $c_2(y)$, $c_2(y)=\pi-c_1(y)$ with
$0\le c_1(y)< \p/2$, we have
\[
\a'(c_1(y))=0=\a'(c_2(y)).
\]
Also we easily compute
\[
\a''(c_i(y))=(1-|l_1|^2)\frac{\cos
\g(c_i(y))}{(\sin\g(c_i(y)))^3}, \quad i=1,2.
\]
Thus, asymptotically,
\[
I_+(n)\sim
I_+^{(1)}(n)e^{\frac{\p}4i}+I_+^{(2)}(n)e^{-\frac{\p}4i}
\]
with
\begin{eqnarray*}
I_+^{(j)}(n)&=&\frac1{\sqrt{n\p}}\frac1{\sqrt{1-|l_1|^2}}|\sin\g(c_1(y))||\tan\g(c_1(y))|^{1/2}\\
&&\hskip .5 true cm \times e^{in(\g(c_j(y))-c_j(y)y)}l_+(c_j(y)),
\quad j=1,2.
\end{eqnarray*}
Also for those $x$ and $n$ satisfying \eqref{eq:weights}
\[
I_-(n)\sim
I_-^{(1)}(n)e^{-\frac{\p}4i}+I_-^{(2)}(n)e^{+\frac{\p}4i}
\]
with (we use symmetry of $\g$)
\begin{eqnarray*}
I_-^{(j)}(n)&=&\frac1{\sqrt{n\p}}\frac1{\sqrt{1-|l_1|^2}}|\sin\g(c_1(y))||\tan\g(c_1(y))|^{1/2}\\
&&\hskip .5 true cm \times
e^{-in(\g(c_j(y))-c_j(y)y)}l_-(-c_j(y)), \quad j=1,2.
\end{eqnarray*}
Similarly we can compute the asymptotics of $J_{\pm}(n)$. Under
the condition \eqref{eq:weights} we have
\[
J_+(n)\sim
J_+^{(1)}(n)e^{\frac{\p}4i}+J_+^{(2)}(n)e^{-\frac{\p}4i}
\]
with
\begin{eqnarray*}
J_+^{(j)}(n)&=&\frac1{\sqrt{n\p}}\frac1{\sqrt{1-|l_1|^2}}|\sin\g(c_1(y))||\tan\g(c_1(y))|^{1/2}\\
&&\hskip .5 true cm \times e^{in(\g(c_j(y))-c_j(y)y)}m_+(c_j(y)),
\quad j=1,2.
\end{eqnarray*}
And
\[
J_-(n)\sim
J_-^{(1)}(n)e^{-\frac{\p}4i}+J_-^{(2)}(n)e^{+\frac{\p}4i}
\]
with
\begin{eqnarray*}
J_-^{(j)}(n)&=&\frac1{\sqrt{n\p}}\frac1{\sqrt{1-|l_1|^2}}|\sin\g(c_1(y))||\tan\g(c_1(y))|^{1/2}\\
&&\hskip .5 true cm \times
e^{-in(\g(c_j(y))-c_j(y)y)}m_-(-c_j(y)), \quad j=1,2.
\end{eqnarray*}
We now apply these asymptotic estimates to
\eqref{eq:characteristic_function2}. Then as a Riemann integral,
the sum over $x\in \mbZ$ becomes an integral over $y$. Moreover,
by Lemma \ref{lem:stationary_phase}, since the leading term
appears at the points that satisfy $\a'=0$, we see from
\eqref{eq:alpha(k)} that the integral over $y$ is supported on the
range of $\a'$, which is $[-|l_1|,|l_1|]$. Finally, by using
Riemann-Lebesgue lemma, we see that the characteristic function
has the limit:
\begin{equation}\label{eq:limit}
\lim_{n\to \infty}\ph_n^{(U;\ps_0)}(\x)=\int e^{i\x
y}\r^{(U;\ps_0)}(y)dy,
\end{equation}
where the density function $\r^{(U:\ps_0)}(y)$ is supported in
$[-|l_1|,|l_1|]$ and is represented by
\begin{equation}\label{eq:density2}
\r^{(U;\ps_0)}(y)=\frac1{\p(1-|l_1|^2)}\sin^2\g(c_1(y))|\tan\g(c_1(y))|g^{(U;\ps_0)}(y),
\end{equation}
with
\begin{eqnarray}\label{eq:intl_dep}
 &&g^{(U;\ps_0)}(y)\nonumber\\
 &&\qquad=\left\{|l_+(c_1(y))|^2+|l_+(c_2(y))|^2+|l_-(-c_1(y))|^2+|l_-(-c_2(y))|^2\right.\\
&&\qquad+\left.|m_+(c_1(y))|^2+|m_+(c_2(y))|^2+|m_-(-c_1(y))|^2+|m_-(-c_2(y))|^2\right\}.\nonumber
\end{eqnarray}
Let us now compute the the factor in the density that does not
depend on the initial condition. By differentiating
\eqref{eq:gamma(k)} and from the definition of $c_1(y)$ we have
\begin{equation}\label{eq:value_at_c_1(y)}
|l_1|\sin c_1(y)=y\sin\g(c_1(y)).
\end{equation}
By \eqref{eq:gamma(k)} and \eqref{eq:value_at_c_1(y)} we get
\begin{equation}\label{eq:trigonometry_at_c_1(y)}
\sin^2\g(c_1(y))=\frac{1-|l_1|^2}{1-y^2} \text{ and } \quad
\cos^2\g(c_1(y))=\frac{|l_1|^2-y^2}{1-y^2}.
\end{equation}
Inserting these into \eqref{eq:density2} we get the first half
part in the density \eqref{eq:density}. The remaining part that
depends on the initial condition is obtained by direct
computation. We have represented the values of $l_{\pm}(\pm
c_j(y))$ and  $m_{\pm}(\pm c_j(y))$ for $j=1,2$ in Lemma
\ref{lem:lm_values} below. By this we get the remaining part
$g^{(U;\ps_0)}(y)$ in \eqref{eq:intl_dep} and the proof for the
case $l_1l_2r_1r_2\neq 0$ is completed.

The case that $l_1=0$ or $l_2=0$. In this case the behaviour of
QW is very simple. We can directly compute the distribution of
$X_n^{(U;\ps_0)}$ from the defining relation
\eqref{eq:evolution1}.
 Let $\ps_0=\left\{\left(\begin{matrix}\ps_0(1;x)\\\ps_0(2;x)\end{matrix}\right)\right\}_{x\in\mbZ}$ be the initial condition. We first consider the case $l_2=0$. Then,
 at time $n$, we have
 \[
 \left(\begin{matrix}\ps_n(1;x)\\\ps_n(2;x)\end{matrix}\right)=\left(\begin{matrix}l_1^n\ps_0(1;x+n)\\r_2^n\ps_0(2;x-n)\end{matrix}\right).
 \]
 Therefore
 \[
 \mbP(X_n^{(U;\ps_0)}=x)=|\ps_0(1;x+n)|^2+|\ps_0(2;x-n)|^2,
 \]
 and hence
 \begin{eqnarray*}
 \mbE(e^{i\x X_n^{(U;\ps_0)}})&=&\sum_{x\in \mbZ}e^{i\x x}\left(|\ps_0(1;x+n)|^2+|\ps_0(2;x-n)|^2\right)\\
 &=&e^{-i\x n}\sum_{x\in \mbZ}e^{i\x x}|\ps_0(1;x)|^2+e^{i\x n}\sum_{x\in \mbZ}e^{i\x x}|\ps_0(2;x)|^2.
 \end{eqnarray*}
Thus, by dominated convergence theorem,  we have
\[
\lim_{n\to \infty}\mbE(e^{i\x
X_n^{(U;\ps_0)}/n})=e^{-i\x}\sum_{x\in
\mbZ}|\ps_0(1;x)|^2+e^{i\x}\sum_{x\in \mbZ}|\ps_0(2;x)|^2.
\]
We conclude that for $l_2=0$ the limit distribution is
\[
\m^{(U;\ps_0)}=(\sum_{x\in \mbZ}|\ps_0(1;x)|^2)
\d_{-1}+(\sum_{x\in \mbZ}|\ps_0(2;x)|^2)\d_1.
\]

Next we consider the case $l_1=0$. Then,
 at time $n$, we have
 \[
 \left(\begin{matrix}\ps_n(1;x)\\\ps_n(2;x)\end{matrix}\right)
 =\begin{cases}(l_2r_1)^{m-1}\left(\begin{matrix}l_2\ps_0(2;x+1)\\r_1\ps_0(1;x-1)\end{matrix}\right),&\text{if}~ n=2m-1\\
  (l_2r_1)^{m}\left(\begin{matrix}\ps_0(1;x)\\\ps_0(2;x)\end{matrix}\right),&\text{if}~ n=2m.\end{cases}
 \]
 Therefore
 \[
 \mbP(X_n^{(U;\ps_0)}=x)=\begin{cases}|\ps_0(1;x-1)|^2+|\ps_0(2;x+1)|^2&\text{if}~~n~\text{is odd}\\
 |\ps_0(1;x)|^2+|\ps_0(2;x)|^2&\text{if}~~n~\text{is even}\end{cases}
 \]
 and hence
 \begin{eqnarray*}
 &&\mbE(e^{i\x X_n^{(U;\ps_0)}})\\
 &&~=\begin{cases} \sum_{x\in \mbZ}e^{i\x x}\left(|\ps_0(1;x-1)|^2+|\ps_0(2;x+1)|^2\right)&\text{if}~~n~\text{is odd}\\
 \sum_{x\in \mbZ}e^{i\x x}(|\ps_0(1;x)|^2+|\ps_0(2;x)|^2)&\text{if}~~n~\text{is even}\end{cases}\\
 &&~=\begin{cases} e^{i\xi}\sum_{x\in \mbZ}e^{i\x x}|\ps_0(1;x)|^2+ e^{-i\xi}\sum_{x\in \mbZ}e^{i\x x}|\ps_0(2;x)|^2&\text{if}~~n~\text{is odd}\\
 \sum_{x\in \mbZ}e^{i\x x}(|\ps_0(1;x)|^2+|\ps_0(2;x)|^2)&\text{if}~~n~\text{is even}.\end{cases}
 \end{eqnarray*}
By dominated convergence theorem again, we have
\[
\lim_{n\to \infty}\mbE(e^{i\x X_n^{(U;\ps_0)}/n})=\sum_{x\in
\mbZ}(|\ps_0(1;x)|^2+|\ps_0(2;x)|^2)=1.
\]
We conclude that for $l_1=0$ the limit distribution is
\[
\m^{(U;\ps_0)}=\d_0.
\]
The proof is completed. $\square$

{\it Proof of Proposition \ref{prop:local_intl_cond}.} If the
particle is located at the origin with a chiral state
$\left(\begin{matrix}a\\b\end{matrix}\right)$, then the Fourier
transform of it is just a constant:
\begin{equation}\label{eq:constant_F}
\wh{\ps}_0(k)\equiv\left(\begin{matrix}\wh{\ps}_0(1;k)\\\wh{\ps}_0(2;k)\end{matrix}\right)=\frac1{\sqrt{2\p}}\left(\begin{matrix}a\\b\end{matrix}\right).
\end{equation}
By using this and Lemma \ref{lem:lm_values} we can directly
compute the function $g^{(U;\ps_0)}(y)$ in \eqref{eq:intl_dep},
which gives exactly the factor $(1-\b^{(U;\ps_0)}y)$
 in the statement of the proposition. By Theorem \ref{thm:asymptotics} the proof is
 completed. $\square$

Now we present the values of functions that are used to get
$g^{(U;\ps_0)}(y)$ in Theorem \ref{thm:asymptotics}, i.e., the
part of limit density function that depends on the initial
conditions. It is obtained by directly computing $l_{\pm}(\pm
c_j(y))$ and $m_{\pm}(\pm c_j(y))$ for $j=1,2$. For this we first
need to compute $S(k)^{-1}$ at $k=\pm c_j(y)$, $j=1,2$.
\begin{lem}\label{lem:S^{-1}}
Suppose that $l_1l_2r_1r_2\neq 0$.
 The values $S(k)^{-1}$ at $k=\pm c_j(y)$, $j=1,2$, are as follows:
 \begin{eqnarray*}
 S(c_1(y))^{-1}&=&\frac{l_2e^{i(c_1(y)-\th_1)}}{2}\left(\begin{matrix}\frac{1-y}{l_2e^{-i\th_1}}&\frac{-1}{|l_1|}\left(y+i\frac{\sqrt{|l_1|^2-y^2}}
 {\sqrt{1-|l_1|^2}}\right)\\\frac{1+y}{l_2e^{-i\th_1}}&\frac{1}{|l_1|}\left(y+i\frac{\sqrt{|l_1|^2-y^2}}
 {\sqrt{1-|l_1|^2}}\right)\end{matrix}\right)\\
 S(c_2(y))^{-1}&=&\frac{l_2e^{i(c_2(y)-\th_1)}}{2}\left(\begin{matrix}\frac{1-y}{l_2e^{-i\th_1}}&\frac{-1}{|l_1|}\left(y-i\frac{\sqrt{|l_1|^2-y^2}}
 {\sqrt{1-|l_1|^2}}\right)\\\frac{1+y}{l_2e^{-i\th_1}}&\frac{1}{|l_1|}\left(y-i\frac{\sqrt{|l_1|^2-y^2}}
 {\sqrt{1-|l_1|^2}}\right)\end{matrix}\right)\\
 S(-c_1(y))^{-1}&=&\frac{l_2e^{-i(c_1(y)+\th_1)}}{2}\left(\begin{matrix}\frac{1+y}{l_2e^{-i\th_1}}&\frac{1}{|l_1|}\left(y+i\frac{\sqrt{|l_1|^2-y^2}}
 {\sqrt{1-|l_1|^2}}\right)\\\frac{1-y}{l_2e^{-i\th_1}}&\frac{-1}{|l_1|}\left(y+i\frac{\sqrt{|l_1|^2-y^2}}
 {\sqrt{1-|l_1|^2}}\right)\end{matrix}\right)\\
 S(-c_2(y))^{-1}&=&\frac{l_2e^{-i(c_2(y)+\th_1)}}{2}\left(\begin{matrix}\frac{1+y}{l_2e^{-i\th_1}}&\frac{1}{|l_1|}\left(y-i\frac{\sqrt{|l_1|^2-y^2}}
 {\sqrt{1-|l_1|^2}}\right)\\\frac{1-y}{l_2e^{-i\th_1}}&\frac{-1}{|l_1|}\left(y-i\frac{\sqrt{|l_1|^2-y^2}}
 {\sqrt{1-|l_1|^2}}\right)\end{matrix}\right)
 \end{eqnarray*}
 \end{lem}
 \begin{proof} We use the definition of $S(k)$ by using the eigenvectors of $U(k)$ in \eqref{eq:eigenvectors} and compute the values at $\pm c_j(y)$, $j=1,2$, as it was done in
 \eqref{eq:trigonometry_at_c_1(y)}. \end{proof}
 It is then straightforward to compute $l_{\pm}(\pm c_j(y))$ and $m_{\pm}(\pm c_j(y))$. Notice that the Fourier transform of the initial vector is denoted by
 $\wh{\ps}_0=\left\{\left(\begin{matrix}\wh{\ps}_0(1;k)\\\wh{\ps}_0(2;k)\end{matrix}\right)\right\}_{k\in\mbK}$.
 \begin{lem}\label{lem:lm_values}
 Suppose that $l_1l_2r_1r_2\neq 0$.
 The values of $l_{\pm}(\pm c_j(y))$ and $m_{\pm}(\pm c_j(y))$, $j=1,2$, are as follows.
 \begin{eqnarray*}
 l_+(c_1(y))&=&\frac{l_2e^{-i \th_1}}{2}\Big(\frac{1-y}{l_2e^{-i\th_1}}\wh{\ps}_0(1;c_1(y)+\th_1)\\
 &&\hskip 2true cm-\frac1{|l_1|}\Big(y+i\sqrt{\frac{|l_1|^2-y^2}{1-|l_1|^2}}\Big)
 \wh{\ps}_0(2;c_1(y)+\th_1)\Big)\\
 l_+(c_2(y))&=&\frac{l_2e^{-i  \th_1}}{2}\Big(\frac{1-y}{l_2e^{-i \th_1}}\wh{\ps}_0(1;c_2(y)+\th_1)\\
  &&\hskip 2true cm-\frac1{|l_1|}\Big(y-i\sqrt{\frac{|l_1|^2-y^2}{1-|l_1|^2}}\Big)
 \wh{\ps}_0(2;c_2(y)+\th_1)\Big)\\
  l_-(-c_1(y))&=&\frac{l_2e^{-i  \th_1}}{2}\Big(\frac{1-y}{l_2e^{-i
 \th_1}}\wh{\ps}_0(1;-c_1(y)+\th_1)\\
 &&\hskip 2true cm -\frac1{|l_1|}\Big(y+i\sqrt{\frac{|l_1|^2-y^2}{1-|l_1|^2}}\Big)
 \wh{\ps}_0(2;-c_1(y)+\th_1)\Big)\\
 l_-(-c_2(y))&=&\frac{l_2e^{-i  \th_1}}{2}\Big(\frac{1-y}{l_2e^{-i \th_1}}\wh{\ps}_0(1;-c_2(y)+\th_1)\\
  &&\hskip 2true cm-\frac1{|l_1|}\Big(y-i\sqrt{\frac{|l_1|^2-y^2}{1-|l_1|^2}}\Big)
 \wh{\ps}_0(2;-c_2(y)+\th_1)\Big)
\end{eqnarray*}
 \begin{eqnarray*}
 m_+(c_1(y))&=&\frac{1-|l_1|^2}{2|l_1|}\Big(\frac{-1}{l_2e^{-i \th_1}}\Big(y-i\sqrt{\frac{|l_1|^2-y^2}{1-|l_1|^2}}\Big)
 \wh{\ps}_0(1;c_1(y)+\th_1)\\
 &&\hskip 2 true cm +\frac{|l_1|}{1-|l_1|^2}(1+y)\wh{\ps}_0(2;c_1(y)+\th_1)\Big)\\
 m_+(c_2(y))&=&\frac{1-|l_1|^2}{2|l_1|}\Big(\frac{-1}{l_2e^{-i \th_1}}\Big(y+i\sqrt{\frac{|l_1|^2-y^2}{1-|l_1|^2}}\Big)
 \wh{\ps}_0(1;c_2(y)+\th_1)\\
 &&\hskip 2 true cm+\frac{|l_1|}{1-|l_1|^2}(1+y)\wh{\ps}_0(2;c_2(y)+\th_1)\Big)\\
 \end{eqnarray*}
 \begin{eqnarray*}
 m_-(-c_1(y))&=&\frac{1-|l_1|^2}{2|l_1|}\Big(\frac{-1}{l_2e^{-i \th_1}}\Big(y-i\sqrt{\frac{|l_1|^2-y^2}{1-|l_1|^2}}\Big)
 \wh{\ps}_0(1;-c_1(y)+\th_1)\\
 &&\hskip 2 true cm+\frac{|l_1|}{1-|l_1|^2}(1+y)\wh{\ps}_0(2;-c_1(y)+\th_1)\Big)\\
 m_-(-c_2(y))&=&\frac{1-|l_1|^2}{2|l_1|}\Big(\frac{-1}{l_2e^{-i \th_1}}\Big(y+i\sqrt{\frac{|l_1|^2-y^2}{1-|l_1|^2}}\Big)
 \wh{\ps}_0(1;-c_2(y)+\th_1)\\
 &&\hskip 2 true cm+\frac{|l_1|}{1-|l_1|^2}(1+y)\wh{\ps}_0(2;-c_2(y)+\th_1)\Big).
 \end{eqnarray*}
 \end{lem}

\end{document}